\documentclass[sigconf]{acmart}
\usepackage{booktabs} 
\usepackage{graphicx}
\usepackage[show]{chato-notes}
\usepackage{amsmath,amssymb}
\usepackage{epsfig}
\usepackage{algorithm}
\usepackage[noend]{algorithmic}
\usepackage{url}
\usepackage{hyperref}
\usepackage{xspace}
\usepackage{tikz}
\usetikzlibrary{shadows}
\usepackage{xcolor}
\usepackage[framemethod=TikZ]{mdframed}
\usepackage{tkz-berge,float}

\makeatletter
\newsavebox\myboxA
\newsavebox\myboxB
\newlength\mylenA

\newcommand*\xoverline[2][0.75]{%
    \sbox{\myboxA}{$\m@th#2$}%
    \setbox\myboxB\null
    \ht\myboxB=\ht\myboxA%
    \dp\myboxB=\dp\myboxA%
    \wd\myboxB=#1\wd\myboxA
    \sbox\myboxB{$\m@th\overline{\copy\myboxB}$}
    \setlength\mylenA{\the\wd\myboxA}
    \addtolength\mylenA{-\the\wd\myboxB}%
    \ifdim\wd\myboxB<\wd\myboxA%
       \rlap{\hskip 0.5\mylenA\usebox\myboxB}{\usebox\myboxA}%
    \else
        \hskip -0.5\mylenA\rlap{\usebox\myboxA}{\hskip 0.5\mylenA\usebox\myboxB}%
    \fi}
\makeatother

\hypersetup{
    bookmarks=true,         
    unicode=false,          
    pdftoolbar=true,        
    pdfmenubar=true,        
    pdffitwindow=false,     
    pdfstartview={FitH},    
    pdftitle={My title},    
    pdfauthor={Author},     
    pdfsubject={Subject},   
    pdfcreator={Creator},   
    pdfproducer={Producer}, 
    pdfnewwindow=true,      
    colorlinks=true,       
    linkcolor=black,          
    citecolor=black,        
    filecolor=black,      
    urlcolor=black           
}

\newcommand{\spara}[1]{\smallskip\noindent{\bf #1}}

\newtheorem{problem}{Problem}

\newcommand{\bigO}{\ensuremath{\mathcal{O}}\xspace}
\newcommand{\bigTheta}{\ensuremath{{\Theta}}\xspace}
\newcommand{\bigOmega}{\ensuremath{{\Omega}}\xspace}

\newcommand{\cc}{\ensuremath{\mathit{cc}}}
\newcommand{\ccad}{\ensuremath{\xoverline{\mathit{cc}}}}

\newcommand{\NP}{$\mathbf{NP}$\xspace}
\newcommand{\NPhard}{$\mathbf{NP}$-hard\xspace}

\newcommand{\eigensign}{{\sc Eigen\-sign}\xspace}
\newcommand{\randalgo}{{\sc Ran\-dom-Eigen\-sign}\xspace}
\newcommand{\simplealgo}{{\sc Pick-an-edge}\xspace}
\newcommand{\wma}{{\ensuremath W_{\mathit{2CC}} }\xspace}

\newcommand{\wsqp}{{\ensuremath W_{\mathit{HPC}}}\xspace}

\newcommand{\twocorrelationclustering}{{\sc 2-Correlation-Clustering}\xspace}
\newcommand{\twocc}{{\sc 2CC}\xspace}
\newcommand{\twoccfull}{{\sc 2CC-Full}\xspace}
\newcommand{\ourproblem}{{\sc 2PC}\xspace}
\newcommand{\ourproblemfull}{{\sc 2-Polarized-Communities}\xspace}
\newcommand{\score}{{polarity}\xspace}

\renewcommand{\vec}[1]{\mathbf{#1}}

\newcommand{\squishlist}{
 \begin{list}{$\bullet$}
  {  \setlength{\itemsep}{0pt}
     \setlength{\parsep}{3pt}
     \setlength{\topsep}{3pt}
     \setlength{\partopsep}{0pt}
     \setlength{\leftmargin}{2em}
     \setlength{\labelwidth}{1.5em}
     \setlength{\labelsep}{0.5em}
} }
\newcommand{\squishlisttight}{
 \begin{list}{$\bullet$}
  { \setlength{\itemsep}{0pt}
    \setlength{\parsep}{0pt}
    \setlength{\topsep}{0pt}
    \setlength{\partopsep}{0pt}
    \setlength{\leftmargin}{2em}
    \setlength{\labelwidth}{1.5em}
    \setlength{\labelsep}{0.5em}
} }

\newcommand{\squishdesc}{
 \begin{list}{}
  {  \setlength{\itemsep}{0pt}
     \setlength{\parsep}{3pt}
     \setlength{\topsep}{3pt}
     \setlength{\partopsep}{0pt}
     \setlength{\leftmargin}{1em}
     \setlength{\labelwidth}{1.5em}
     \setlength{\labelsep}{0.5em}
} }

\newcommand{\squishend}{
  \end{list}
}

\copyrightyear{2019}
\acmYear{2019}
\acmConference[CIKM '19]{The 28th ACM International Conference on Information and Knowledge Management}{November 3--7, 2019}{Beijing, China}
\acmBooktitle{The 28th ACM International Conference on Information and Knowledge Management (CIKM '19), November 3--7, 2019, Beijing, China}
\acmPrice{15.00}
\acmDOI{10.1145/3357384.3357977}
\acmISBN{978-1-4503-6976-3/19/11}

\begin{document}
\fancyhead{}
\title{Discovering Polarized Communities in Signed Networks}

\author{Francesco Bonchi}
\email{francesco.bonchi@isi.it}
\orcid{1234-5678-9012}
\affiliation{%
  \institution{ISI Foundation, Italy}
  \institution{Eurecat, Barcelona, Spain}
}

\author{Edoardo Galimberti}
\email{edoardo.galimberti@unito.it}
\affiliation{%
  \institution{Dept. of Computer Science., University of Turin, Italy}
  \institution{ISI Foundation, Italy}
}

\author{Aristides Gionis, Bruno Ordozgoiti}
\email{{name.surname}@aalto.fi}
\affiliation{%
  \institution{Dept. of Computer Science, Aalto University, Finland}
}

\author{Giancarlo Ruffo}
\email{giancarlo.ruffo@unito.it}
\affiliation{%
  \institution{Dept. of Computer Science., University of Turin, Italy}
}

\renewcommand{\shortauthors}{Bonchi, et al.}

\begin{abstract}
Signed networks contain edge annotations to indicate whether each interaction is friendly (positive edge) or antagonistic (negative edge). The model is simple but powerful and 
it can capture novel and interesting structural properties of real-world phenomena.
The analysis of signed networks has many applications from modeling discussions in social media, to mining user reviews, and to recommending products in e-commerce sites.

In this paper we consider the problem of discovering polarized communities in signed networks.
In particular, we search for two communities (subsets of the network vertices) where within communities there are mostly positive edges while across communities there are mostly negative edges.
%
We formulate this novel problem as a ``discrete eigenvector'' problem, which we show to be \NP-hard. We then develop two intuitive spectral algorithms: one deterministic, and one randomized with quality guarantee $\sqrt{n}$ (where $n$ is the number of vertices in the graph), tight up to constant factors.

We validate our algorithms against non-trivial baselines on real-world signed networks.
Our experiments confirm that our algorithms produce higher quality solutions, are much faster and can scale to much larger networks than the baselines, and are able to detect ground-truth polarized communities.

\end{abstract}

\maketitle
\sloppy

\section{Introduction}
\label{sec:intro}
The increase of polarization around controversial issues is a growing concern
with important societal fallouts.
While controversy can be engaging, and can lead to users spending more time on social-media platforms,
in disproportionate amounts it can generate a negative user experience,
potentially leading to the abandonment of the platform.
Excessive polarization, together with the emergence of bots and the spread of misinformation, has thus become an urgent technological problem that needs to be solved.
It is not surprising that the last few years have witnessed an uptake of the research on methods for the detection and suppression of these phenomena \cite{garimella2017reducing,liao2014can,liao2014expert,vydiswaran2015overcoming,munson2013encouraging,graells2014people}.


\begin{figure}[t]
\centerline{
\includegraphics[width=.95\columnwidth]{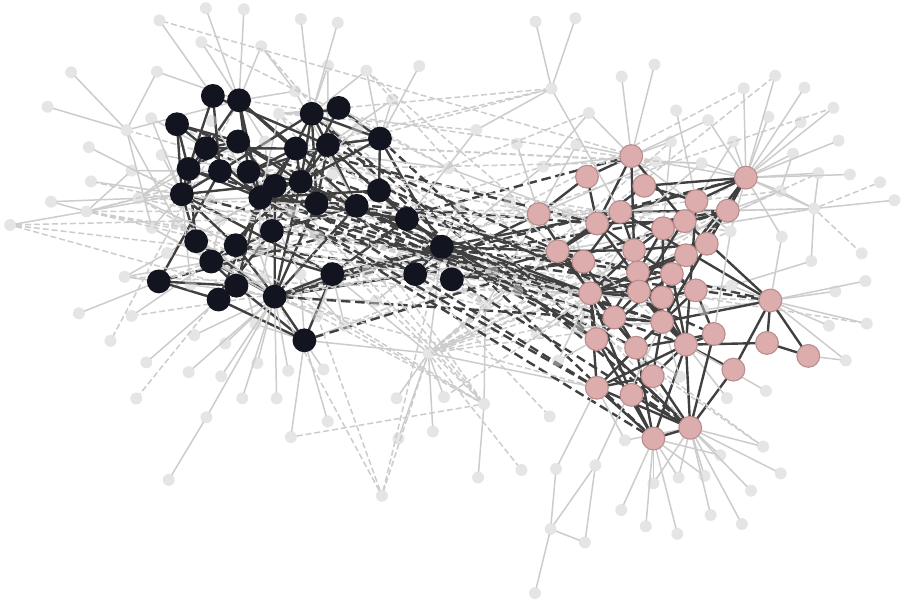}
}
\vspace{-5mm}
\caption{\label{fig:example} An example of two polarized communities in the \textsf{Congress} network (dataset details in Section \ref{sec:experiments}).
Solid edges are positive, while dashed edges are negative.
}
\end{figure}

While polarization is a well studied phenomenon in political and social sciences~\cite{baldassarri2007dynamics,brundidge2010encountering,esteban1994measurement,feldman2014mutual,garrett2014partisan,wojcieszak2009online},
modern social-media platforms brought it to a different scale,
providing an unprecedented wealth of data.
The necessity to analyze the available data and gain valuable insights brings new algorithmic challenges.

In order to study polarization in large-scale online data,
one first step is to \emph{detect} it. As a step in this direction, in this paper we study a fundamental problem abstraction for this task,
i.e., the problem of \emph{discovering polarized communities in signed networks}.

A signed network is a simple, yet general and powerful, representation:
vertices represent entities and edges between vertices represent interactions,
which can be friendly (positive) or antagonistic (negative)~\cite{harary1953notion}.
Signed graphs analysis has many applications
from modeling interactions in social media~\cite{kunegis2009slashdot},
to mining user reviews~\cite{beigi2016signed},
to studying information diffusion and epidemics~\cite{li2013influence},
to recommending products in e-commerce sites~\cite{ma2009learning,victor2011trust}, and to estimating the structural balance of a (physical) complex system~\cite{antal2006,marvel2009}.

In this paper, we introduce the \ourproblemfull\ problem (\ourproblem),
which requires finding two communities (subsets of the network vertices)
such that within communities there are mostly positive edges
while across communities there are mostly negative edges.
Furthermore, we do not aim to partition the whole network,
so the two polarized communities we are searching can be concealed
within a large body of other network vertices,
which are neutral with respect to the polarized structure.
Our hypothesis is that such 2-community polarized structure accurately
captures controversial discussions in real-world social-media platforms.

Figure \ref{fig:example} shows an example of the two most polarized communities found in the \textsf{Congress} network (details in Section \ref{sec:experiments}). The two communities involve $34$ and $37$ vertices (out of $219$), respectively, having more than 98\% of positive edges within and 78\% negative edges across.
The vertices in gray do not participate in any of the two polarized communities: either they have too few connections with any community, or the polarity of their relations are mixed and thus their position within the debate unclear.

Our work is, to the best of our knowledge, the first to propose a spectral method for extracting polarized communities from signed networks. In addition, we present hardness results and approximation guarantees.
Our problem formulation deviates from the bulk of the literature where methods
typically look for finding many communities while partitioning the
whole network~\cite{anchuri2012communities,bansal2004correlation,chiang2012scalable,coleman2008local,giotis2006correlation,kunegis2010spectral}.
As discussed in more detail in Sections \ref{sec:related} and \ref{sec:problem}, the closest to our problem statement is the work by Coleman et al.~\cite{coleman2008local},
who employ the correlation-clustering framework and search for exactly two communities.
However, while in that work all vertices must be
included in a cluster, in our setting we allow vertices not to be part
of any cluster. This captures the fact that polarized communities are
typically concealed within a large body of neutral vertices in a
social network. An algorithm that attempts to partition the whole
network would fail to reveal these communities.
As an additional feature, our
methods can be fine-tuned to increase or decrease the
size of the discovered communities. Finally,
while some spectral techniques promote balanced partitions, we
hypothesize that two polarized communities might be of very different
sizes, and thus our problem formulation does not enforce evenly sized
subgraphs.

Our reliance on spectral methods carries several benefits. First, it
is possible to leverage readily available, highly optimized, and
parallelized software implementations. This makes it straightforward
for the practitioner to analyze large networks in real settings using
our approach. Second, even though in this paper we focus on the case
of two communities, we can take inspiration from the existing
literature on spectral graph partitioning to easily extend our algorithms
to the case of an arbitrary number of subgraphs, e.g., by recursive
two-way partitioning or the analysis of multiple
eigenvectors \cite{shi2000normalized}.

In this paper we make the following contributions:
\squishlist

\item We formulate the \ourproblemfull\ problem  (\ourproblem) as a ``discrete eigenvector'' problem (Section \ref{sec:problem}).

\item Exploiting a reduction from classic correlation clustering, we prove that \ourproblem\ is \NP-hard (Theorem~\ref{th:hard}).

\item We devise two intuitive spectral algorithms (Section \ref{sec:algorithms}), one deterministic, and one randomized with quality guarantee $\sqrt{n}$ (Theorem~\ref{the:approx_random}),
which is tight up to constant factors. We believe these to be the first purely combinatorial bounds for spectral methods. Our results apply to graphs of arbitrary weights. Our algorithms' running time is essentially the time required to compute the first eigenvector of the adjacency matrix of the input graph.

\item Our experiments (Section \ref{sec:experiments}) on a large collection of real-world signed networks show that the proposed algorithms discover higher quality solutions, are much faster than the baselines, and can scale to much larger networks. In addition, they are able to identify ground-truth planted polarized communities in synthetic datasets.

\squishend

Related literature is discussed 
in the next section, while Section \ref{sec:conclusions} discusses future work and concludes the paper.

\section{Background and related work}
\label{sec:related}
\noindent \textbf{Signed networks.}
Signed graphs appeared in a work by Harary,
who was interested in the notion of {\em balance} in graphs~\cite{harary1953notion}.
In 1956, Cartwright and Harary generalized Heider's psychological theory
of balance in triangles of sentiments to the theory of balance in
signed graphs~\cite{cartwright1956structural}.

A more recent line of work
develops the spectral properties of signed graphs,
still related to balance theory.
Hou et al.~\cite{hou2003laplacian} prove that a connected signed graph is balanced
if and only if the smallest eigenvalue of the Laplacian is~0.
Hou~\cite{hou2005bounds} also investigates the relationship between the smallest eigenvalue of the Laplacian
and the unbalancedness of a signed graph.

Signed graphs have also been studied
in different contexts.
Guha et al.~\cite{guha2004propagation}
and Les\-ko\-vec et al.~\cite{leskovec2010signed}
study {\em directed} signed graphs
and develop {\em status theory}, 
to reason about the importance of the vertices in such graphs.
Other lines of research include edge and vertex classification~\cite{cesa2012correlation, tang2016node},
link prediction~\cite{leskovec2010predicting,symeonidis2010transitive},
community detection~\cite{ailon2008aggregating,anchuri2012communities,bansal2004correlation,swamy2004correlation},
recommendation~\cite{tang2016recommendations},
and more.
A detailed survey on the topic is provided by Tang et al.~\cite{tang2016survey}.

A few recent works explore the problem of finding
antagonistic communities in signed networks, though with approaches
fundamentally different to ours.
Lo et al.\ consider directed graphs
and search for strongly-connected positive subgraphs that are
negative bi-cliques \cite{lo2011mining}, which severely limits the size
of the resulting communities.
A relaxed variant for undirected networks was
described in subsequent work \cite{gao2016detecting}.
Chu et al.\ propose a constrained-programming objective to find
$k$ warring factions \cite{chu2016finding}, as well as an efficient
algorithm to find local optima.

\spara{Correlation Clustering.}
In the standard correlation-clustering problem~\cite{bansal2004correlation},
we ask to partition the vertices of a signed graph into clusters
so as to maximize (minimize) the number of edges that ``agree'' (``disagree'') with the partitioning,
i.e., the number of positive (negative) edges within clusters
plus the number of negative (positive) edges across clusters.
In the original problem formulations,
such as the ones studied by
Bansal et al.~\cite{bansal2004correlation},
Swamy~\cite{swamy2004correlation}, and
Ailon et al.~\cite{ailon2008aggregating},
the number of clusters is not given as input,
instead it is part of the optimization.
More recent works study the correlation-clustering problem
with additional constraints,
e.g., Giotis and Guruswami~\cite{giotis2006correlation} fix the number of clusters,
Coleman et al.~\cite{coleman2008local} consider only two clusters,
while Puleo and Milenkovic~\cite{puleo2015correlation}
consider constraints on the cluster sizes.

The problem we study
could be seen as
a variant of correlation clustering
where we search for two clusters,
while we allow vertices not to be part of any cluster.

\spara{Detecting polarization in social media.}
A number of papers have studied the problem of detecting polarization in social media
Some approaches are based on text analysis~\cite{choi2010identifying,mejova2014controversy,popescu2010detecting},
while other approaches
consider a graph-theoretic setting~\cite{akoglu2014quantifying,conover2011predicting,garimella2018quantifying}.
However, our work differs significantly from these papers,
as we consider signed networks,
we provide a correlation-clustering problem formulation,
and obtain results with approximation guarantees.

\section{Problem statement}
\label{sec:problem}

Our setting is reminiscent to the {\em correlation-clustering} problem~\cite{bansal2004correlation}, which we recall here.
Given a signed network $G = (V, E_+, E_-)$,
where $E_+$ is the set of positive edges and $E_-$ the set of negative edges,
the goal is to find
a partition of the vertices into $k$ clusters,
so as to
maximize the number of positive edges within clusters plus
the number of negative edges between clusters.

An interesting property of the correlation-clustering formulation
is that one does not need to specify in advance the number of clusters $k$,
instead it is part of the optimization.
In certain cases, however,
the number of clusters is given as input.
The general problem (given $k$) has been studied by
Giotis and Guruswami~\cite{giotis2006correlation},
while Coleman et al.~\cite{coleman2008local}
studied the \twocorrelationclustering\ problem ($k=2$).
The problem arises, for instance,
in the domain of social networks,
where two well-separated clusters reveal a polarized structure.
It can be defined as follows.

\begin{problem}[\twocc]
  \label{problem_maxagree}
   Given a signed network $G = (V,E_+,E_-)$, find a partition $S_1, S_2$ of $V$ so as to maximize
  \begin{align*}
  	\cc(S_1,S_2) =
    \sum_{\substack{i\in\{1,2\}\\(u,v) \in S_i\times S_i}}\hspace{-1em}\frac{1}{2}\mathbf 1_{E_+}(u,v)+\sum_{(u,v)\in S_1\times S_2}\hspace{-1em}\mathbf 1_{E_-}(u,v),
  \end{align*}
where $\mathbf 1_S$ is the indicator function of the set $S$.
\end{problem}

A crucial limitation of the \twocc\ problem is
that all vertices must be accounted for in one of the two clusters.
From an application perspective, however, this may be a strong assumption.
For example, in a social network, we may expect two polarized communities on a topic,
but there may be many individuals who are neutral.

In order to find communities embedded within large networks, we need to exclude neutral vertices from the solution. Therefore, a first approach might be to consider maximizing agreements including a neutral cluster, that is,
finding a partition of $V$ into $S_1$, $S_2$, and $S_0$,
so that $S_1$ and $S_2$ are the two polarized communities,
and $S_0$ is the neutral community,
and the \twocc\ objective $\cc(S_1,S_2)$ is maximized.
However, this modification does not change  the problem significantly.
It is easy to see that it is always no worse
to switch a vertex from cluster $S_0$ to one of the other two clusters.

\begin{proposition}
Let $S_0, S_1, S_2$ be any partition of $V$, with $S_0\not=\emptyset$.
Then there is always a partition $S_1', S_2'$ of $V$
(i.e., $S_1'\cup S_2' =V$ and $S_1'\cap S_2' =\emptyset$)
with $S_1\subseteq S_1'$ and $S_2\subseteq S_2'$ so that
\[
\cc(S_1',S_2') \ge \cc(S_1,S_2).
\]
\end{proposition}

A further modification might be to subtract disagreements from the value of the solution,
that is, to maximize agreements minus disagreements.
In other words, we consider the following problem.
\begin{problem}[\twoccfull]
  \label{problem_maxagreedisagree}
  Given a signed network $G = (V,E_+,E_-)$, find a partition $S_0, S_1, S_2$ of $V$ so as to maximize
 
  \begin{eqnarray*}
  	\ccad(S_1,S_2) &  = &
    \sum_{\substack{i\in\{1,2\}\\(u,v) \in S_i\times S_i}}\hspace{-1em}\frac{1}{2}\left({\mathbf 1_{E_+}}(u,v) - {\mathbf 1_{E_-}}(u,v)\right) \\
&  &  + \sum_{(u,v)\in S_1\times S_2}\hspace{-1em} \left({\mathbf 1_{E_-}}(u,v) - {\mathbf 1_{E_+}}(u,v) \right),
  \end{eqnarray*}
where $\mathbf 1_S$ is the indicator function of the set $S$.
\end{problem}

Unfortunately, problem \twoccfull\ suffers from the same issue as problem \twocc:
switching a vertex from the neutral cluster $S_0$ to one of the polarized clusters $S_1$ or $S_2$
(the one that is best) leads to no worse solution according to the objective~\ccad.

\begin{proposition}
Let $S_0, S_1, S_2$ be any partition of $V$, with $S_0\not=\emptyset$.
Then there is always a partition $S_1', S_2'$ of $V$
(i.e., $S_1'\cup S_2' =V$ and $S_1'\cap S_2' =\emptyset$)
with $S_1\subseteq S_1'$ and $S_2\subseteq S_2'$ so that
\[
\ccad(S_1',S_2') \ge \ccad(S_1,S_2).
\]
\end{proposition}

A nice property of the  \ccad\ objective
is that it can be written neatly in a matrix notation.
Let $A$ be the adjacency matrix of the signed network $G = (V,E_+,E_-)$,
where
positive edges $(i,j)\in E_+$ are indicated by $A_{ij}=1$,
negative edges $(i,j)\in E_-$ are indicated by $A_{ij}=-1$,
and
non-edges are indicated by $A_{ij}=0$.
A partition $S_0, S_1, S_2$ of~$V$
can be represented by a vector $\vec{x}\in\{-1,0,1\}^n$,
whose $i$-th coordinate is
$x_i=0$ if $i\in S_0$,
$x_i=1$ if $i\in S_1$,
and
$x_i=-1$ if $i\in S_2$.
Then \twoccfull\ can be reformulated as follows.

\begin{problem}[\twoccfull]
\label{problem_maxagreedisagree_matrix}
Given a signed network $G = (V, E_+, E_-)$ with $n$ vertices and signed adjacency matrix $A$,
find a partition $S_0, S_1, S_2$ of $V$ represented by vector $\vec{x}\in\{-1,0,1\}^n$
maximizing
\[
\ccad(S_1,S_2) = \vec{x}^T A\, \vec{x}.
\]
\end{problem}

Since our goal is to discover polarized communities
$S_1$ and $S_2$ that are potentially concealed within other neutral vertices $S_0$,
we want to find minimal sets $S_1$ and $S_2$.
This can be achieved by normalizing $\vec{x}^T A\, \vec{x}$
with the size of $S_1$ and $S_2$, which in vector form is $\vec{x}^T \vec{x}$.
This consideration motivates our last problem formulation,
which we dub \ourproblemfull\ (\ourproblem).

\smallskip
\begin{mdframed}[innerbottommargin=4pt,innertopmargin=0pt,innerleftmargin=5pt,innerrightmargin=5pt,backgroundcolor=gray!10,roundcorner=10pt]
\begin{problem}[\ourproblem]
\label{problem_sqp}
Given a signed network $G = (V, E_+, E_-)$ with $n$ vertices and signed adjacency matrix $A$, find
a vector $\vec{x}\in\{-1,0,1\}^n$ that maximizes
\[
\frac{\vec{x}^T A\, \vec{x}}{\vec{x}^T \vec{x}}.
\]
\end{problem}
\end{mdframed}

In the rest of this paper we refer to the objective function of Problem \ref{problem_sqp} as \emph{\score}.
As \score is penalized with the size of the solution, vertices are only added to one of the two clusters if they contribute significantly to the objective.
We show this problem to be \NPhard\ (proof in the Appendix) and propose algorithms with approximation guarantees.

\begin{theorem}\label{th:hard}
  \ourproblem is \NPhard.
\end{theorem}
%


It should be noted that \ourproblem does not enforce balance between the communities.
This can be beneficial if there exist polarized communities of significantly different size in the input network. In an extreme case,
the solution could even be comprised of a single cluster if there is a large, dense community that overwhelms any other polarized formation.

\section{Algorithms}
\label{sec:algorithms}
The formulation of \ourproblem\ is suggestive of spectral theory,
which we utilize to design our algorithms.
We propose and analyze two spectral algorithms:
one is deterministic,
while the second is randomized and achieves approximation guarantee $\sqrt{n}$.
The running time of both algorithms is dominated by the computation of a spectral decomposition of the adjacency matrix. In practice, this can be done using readily available implementations that exploit sparsity and can run in parallel on multiple cores.

The first algorithm, \eigensign, works by simply discretizing the entries of the eigenvector of the adjacency matrix corresponding to the largest eigenvalue.

%

To illustrate the difficulty of approximating \ourproblem, we analyze the following simple algorithm, which we refer to as \simplealgo. Pick an arbitrary edge: if it is positive, put the endpoints in one cluster, leaving the other cluster empty; if it is negative, put the endpoints in separate clusters.
\begin{proposition}
  \label{the:approx_eigensign}
The \simplealgo algorithm gives an $n$-appro\-xi\-ma\-tion of the optimum.
\end{proposition}
\begin{proof}
  The described algorithm outputs a solution $\vec{x}$ such that
  \[
  \frac{\vec{x}^T A\, \vec{x}}{\vec{x}^T \vec{x}} \geq 1.
  \]
  The result now follows from the fact that
  $OPT \leq \lambda_1 \leq n$,
  where $\lambda_1$ is the largest eigenvalue of $A$.\end{proof}

In the case of networks with arbitrary real weights, it can be shown that despite the close relationship between the \ourproblem objective and the leading eigenvector of $A$, \eigensign cannot do better than this up to constant factors. Consider a fully connected network with one edge $(u,v)$ of weight $w \gg 0$. The rest of the edges have weight close to zero. The primary eigenvector of the adjacency matrix has two entries --- those corresponding to $u$ and $v$ --- of the form $1/\sqrt 2-\epsilon$ for some small $\epsilon$, while the rest are close to zero.
We construct a solution vector $\vec{y}$ as follows:
the two entries corresponding to $u$ and $v$ are set to 1, and the rest to 0.
We have $\vec{y}^TA\,\vec{y}/\vec{y}^T\vec{y}\approx w$.
On the other hand, the \eigensign\ algorithm outputs a vector $\vec{x}$ for which
$\vec{x}^TA\,\vec{x}/\vec{x}^T\vec{x}\approx 2w/n$.
It should be noted however, that the focus of this paper is the analysis of the \ourproblem\ problem on signed networks.
The approximation capabilities of the \eigensign\ algorithm on signed networks
(the adjacency matrix $A$ contains entries with values only $-1$, 0, and 1) are left open.

\eigensign generally outputs a solution comprised of all the vertices in the graph --- unless some components of the eigenvector are exactly zero  ---
which is, of course, counter to the motivation of our problem setting.

To overcome this issue we propose a randomized algorithm, \randalgo,
which also computes the first eigenvector, i.e., $\vec{v}$, of the adjacency matrix.
Instead of simply discretizing the entries of $\vec{v}$, it randomly sets each entry of $\vec{x}$ to 1 or -1 with probabilities determined by the entries of $\vec{v}$.
Entries $v_i$ with large magnitude $|v_i|$ are more likely to turn into $\mathit{sgn}(v_i)$ ($-1$ or 1),
while
entries $v_i$ with small magnitude $|v_i|$ are more likely to turn into 0.
For details see Algorithm~\ref{alg:rand_eigensign}.
Note that if $\vec{x}$ is the output of \randalgo, then $\mathbb E[\vec{x}]=\vec{v}$.

The next theorem shows approximation guarantees of \randalgo for signed networks.

\begin{theorem}
  \label{the:approx_random}
Algorithm \randalgo gives a $\sqrt n$-appro\-xi\-ma\-tion of the optimum in expectation.
\end{theorem}

\begin{proof}
First, observe that we can rewrite the expected value of the objective as follows:
\begin{align*}
  \mathbb E \left [\frac{\vec{x}^TA\vec{x}}{\vec{x}^T\vec{x}} \right ] & = \sum_{k=1}^n  \mathbb E \left [\left . \frac{\vec{x}^TA\vec{x}}{\vec{x}^T\vec{x}}\right |\vec{x}^T\vec{x}=k \right ]Pr(\vec{x}^T\vec{x}=k)
  \\ & = \sum_{k=1}^n \frac{1}{k} \mathbb E\left[\vec{x}^TA\vec{x}|\vec{x}^T\vec{x}=k\right]Pr(\vec{x}^T\vec{x}=k)
  \\ & = \sum_{k=1}^n \frac{1}{k} \sum_{i\neq j}\mathbb E\left[A_{ij}x_ix_j|\vec{x}^T\vec{x}=k\right]Pr(\vec{x}^T\vec{x}=k).
\end{align*}
If we define $s_{ij}=sgn(v_i)sgn(v_j)$, where $sgn(x)$ denotes the sign of $x\in \mathbb R$, for all $i,j$ we have
\begin{align}
\label{eq:expected_expansion}
  &\mathbb E\left[A_{ij}x_ix_j|\vec{x}^T\vec{x}=k\right]Pr(\vec{x}^T\vec{x}=k) \nonumber
  \\ & = A_{ij}s_{ij}Pr(x_i=1,x_j=1|\vec{x}^T\vec{x}=k)Pr(\vec{x}^T\vec{x}=k).
\end{align}
We now invoke Bayes' theorem and proceed.
\begin{align*}
   &\sum_{k=1}^n \frac{1}{k} \sum_{i\neq j}A_{ij}s_{ij}Pr(x_i=1,x_j=1)Pr(\vec{x}^T\vec{x}=k|x_i=1,x_j=1)
  \\ & = \sum_{k=1}^n \frac{1}{k} \sum_{i\neq j}A_{ij}v_iv_jPr(\vec{x}^T\vec{x}=k|x_i=1,x_j=1)
  \\ & = \sum_{i\neq j} A_{ij}v_iv_j \sum_{k=1}^n \frac{1}{k} Pr(\vec{x}^T\vec{x}=k|x_i=1,x_j=1)
  \\ & = \sum_{i\neq j}A_{ij}v_iv_j \mathbb E\left [\frac{1}{\vec{x}^T\vec{x}}|x_i=1,x_j=1 \right ].
\end{align*}
Since $1/x$ is a convex function, by Jensen's inequality it is
\[
\mathbb E\left [\frac{1}{\vec{x}^T\vec{x}}|x_i=1,x_j=1 \right ] \geq \frac{1}{\mathbb E\left [\vec{x}^T\vec{x}|x_i=1,x_j=1 \right ]}.
\]
Furthermore, for any $i,j,$
\[
\mathbb E\left [\vec{x}^T\vec{x}|x_i=1,x_j=1 \right ] \leq 2+\sqrt{n-2}.
\]
To see this, observe that $\mathbb E\left [\vec{x}^T\vec{x}\right ]=\|\vec{v}\|_1\leq \sqrt n$.
\noindent So we have
\[
\mathbb E\left [\frac{1}{\vec{x}^T\vec{x}}|x_i=1,x_j=1 \right ] \geq \frac{1}{2+\sqrt{n-2}}.
\]
Therefore,
\begin{align*}
  \mathbb E \left [\frac{\vec{x}^TA\vec{x}}{\vec{x}^T\vec{x}} \right ] & = \sum_{i\neq j}A_{ij}v_iv_j \mathbb E\left [\frac{1}{\vec{x}^T\vec{x}}|x_i=1,x_j=1 \right ]
  \\ & \geq \sum_{i\neq j}A_{ij}v_iv_j \frac{1}{2+\sqrt{n-2}}  = \frac{\lambda_1}{2+\sqrt{n-2}}.
\end{align*}
That is,
$$
O(\sqrt n) \mathbb E \left [\frac{\vec{x}^TA\vec{x}}{\vec{x}^T\vec{x}} \right ] \geq \lambda_1 \geq OPT.
$$ \end{proof}

In the appendix we show that this result is tight.

\begin{algorithm}[t]
  \caption{\eigensign}
  Input: adjacency matrix $A$
  \begin{algorithmic}[1]
    \STATE Compute $\vec{v}$, the eigenvector corresponding to the largest eigenvalue $\lambda_1$ of $A$.
    \STATE Construct $\vec{x}$ as follows: for each $i \in \{1, \dots, n\}$, $x_i=\mathit{sgn}(v_i)$.
    \STATE Output $\vec{x}$.
  \end{algorithmic}
  \label{alg:eigensign}
\end{algorithm}
\begin{algorithm}[t]
  \caption{\randalgo}
  Input: adjacency matrix $A$
  \begin{algorithmic}[1]
\STATE Compute $\vec{v}$, the eigenvector corresponding to the largest eigenvalue $\lambda_1$ of $A$.
\STATE Construct $\vec{x}$ as follows: for each $i \in \{1, \dots, n\}$, run a Bernoulli experiment with success probability $|v_i|$. If it succeeds, then $x_i=\mathit{sgn}(v_i)$, otherwise $x_i=0$.
\STATE Output $\vec{x}$.
  \end{algorithmic}
  \label{alg:rand_eigensign}
\end{algorithm}

\subsection{Enhancements for practical use}
\label{subsection:practical}

When using these algorithms to analyze real-world networks in practical applications,
it might be beneficial to apply tweaks to enhance their flexibility and produce a wider variety of results.
We propose the following simple enhancements.

\eigensign: As discussed above, \eigensign always outputs a solution involving all the vertices in the network. We can circumvent this shortcoming by including only those vertices such that the corresponding entry of the eigenvector $\vec{v}$ is at least a user-defined threshold $\tau$. That is, $x_i=\mathit{sgn}(v_i)$ if $|v_i| \geq \tau$, 0 otherwise.

\randalgo:
The $\sqrt{n}$-approximation guaranteed by \randalgo\ is matched in the extreme case in which all entries of the eigenvector $\vec{v}$ are of equal magnitude. Paradoxically, in this situation a solution comprised of all vertices would be optimal, but each vertex is included with a small probability of $1/\sqrt n$. We could of course fix this by modifying the probabilities to be $\min\{1, \sqrt n |v_i|\}$ for each $i$. However, in the opposite extreme, where most of the magnitude of $\vec{v}$ is concentrated in one entry, modifying the probabilities this way might disproportionately boost the likelihood of including undesirable vertices. An adequate multiplicative factor for both cases is $\|\vec{v}\|_1$, modifying the probabilities to be $\min\{1, \|\vec{v}\|_1 |v_i|\}$ for each $i$; in the first case, all vertices are taken with probability 1, while in the second, the probabilities remain almost unchanged. We employed this factor in our experiments with satisfactory results.

An obvious question arising is whether the approximation guarantee of \randalgo  could be improved using the modification described above.
This question is left for future investigation.

\section{Experimental Assessment}
\label{sec:experiments}
This section presents the evaluation of the proposed algorithms:
first (Section~\ref{sec:exp1}) we present a characterization of the polarized communities discovered by our methods; then  (Section~\ref{sec:exp2}) we compare our methods against non-trivial baselines in terms of objective, efficiency and scalability, and ability to detect ground-truth planted polarized communities in synthetic datasets. Finally, we show a case study about political debates (Section~\ref{sec:exp3}).

\spara{Datasets.}
We select publicly-available real-world signed networks, whose main characteristics are summarized in Table~\ref{tab:datasets}.
\textsf{HighlandTribes}\footnote{\href{http://konect.cc}{konect.cc}\label{foot:k}} represents the alliance structure of the Gahuku--Gama tribes of New Guinea.
\textsf{Cloister}$^{\ref{foot:k}}$ contains the esteem/disesteem relations of monks living in a cloister in New England (USA).
\textsf{Congress}$^{\ref{foot:k}}$ reports (un/)favorable mentions of politicians speaking in the US Congress.
\textsf{Bitcoin}\footnote{\href{http://snap.stanford.edu}{snap.stanford.edu}\label{foot:s}} and \textsf{Epinions}$^{\ref{foot:s}}$ are who-trusts-whom networks of the users of Bitcoin OTC and Epinions, respectively.
\textsf{WikiElections}$^{\ref{foot:k}}$ includes the votes about admin elections of the users of the English Wikipedia.
\textsf{Referendum}%
\footnote{\href{https://www.researchgate.net/publication/324517807_Annotated_Corpus_for_Stance_Detection_-_Italian_Constitutional_Referendum_2016}{researchgate.net/publication/324517807\_Annotated\_Corpus\_for\_Stance\_Detection\_-\_Italian\_Constitutional\_Referendum\_2016}}%
~\cite{lai2018stance} records Twitter data about the 2016 Italian Referendum: an interaction is negative if two users are classified with different stances, and is positive otherwise.
\textsf{Slashdot}$^{\ref{foot:s}}$ contains friend/foe links between the users of Slashdot.
The edges of \textsf{WikiConflict}$^{\ref{foot:s}}$ represent positive and negative edit conflicts between the users of the English Wikipedia.
\textsf{WikiPolitics}$^{\ref{foot:k}}$ represents interpreted interactions between the users of the English Wikipedia that have edited pages about politics.

\noindent In order to study scalability, we artificially augment two of the largest datasets to produce networks with millions of vertices and tens of millions of edges (details in Section~\ref{sec:exp2}).

\begin{table}[t!]
\centering
\caption{\label{tab:datasets}Signed networks used:
number of vertices and edges;
ratio of negative edges ($\rho_- = \frac{|E_-|}{|E_+ \cup E_-|}$);
$L_1$-norm of the eigenvector corresponding the largest eigenvalue of $A$ ($\|\vec{v}\|_1$);
 and, ratio of non-zero elements of $A$ ($\delta = \frac{2|E_+ \cup E_-|}{|V|(|V|-1)}$).}
\vspace{-4mm}
\centerline{
\small
\begin{tabular}{lrrrrr}
\toprule
Real-world datasets & $|V|$    &    $|E_+ \cup E_-|$    &    $\rho_-$    &    $\|\vec{v}\|_1$    &    $\delta$\\
\midrule
\textsf{HighlandTribes} &  $16$\,\,\, &  $58$\,\,\,\, & $0.50$ &  $3.61$    &    $0.48$\,\,\,\,\,\,\,\,\,\,\\
\textsf{Cloister}       &  $18$\,\,\, & $125$\,\,\,\, & $0.55$ &  $3.71$    &    $0.81$\,\,\,\,\,\,\,\,\,\,\\
\textsf{Congress}       & $219$\,\,\, & $521$\,\,\,\, & $0.20$    &    $10.51$    &    $0.02$\,\,\,\,\,\,\,\,\,\,\\
\textsf{Bitcoin}        &   $5$\,k    &  $21$\,k\,    & $0.15$    &    $31.21$    &    $1.2e\!-\!03$\\
\textsf{WikiElections}  &   $7$\,k    & $100$\,k\,    & $0.22$    &    $35.96$    &    $3.9e\!-\!03$\\
\textsf{Referendum}     &  $10$\,k    & $251$\,k\,    & $0.05$    &    $42.66$    &    $4.2e\!-\!03$\\
\textsf{Slashdot}       &  $82$\,k    & $500$\,k\,    & $0.23$    &    $59.46$    &    $1.4e\!-\!04$\\
\textsf{WikiConflict}   & $116$\,k    &   $2$\,M      & $0.62$ &$119.66$    &    $2.9e\!-\!04$\\
\textsf{Epinions}       & $131$\,k    & $711$\,k\,    & $0.17$    &    $72.20$    &    $8.2e\!-\!05$\\
\textsf{WikiPolitics}   & $138$\,k    & $715$\,k\,    & $0.12$    &    $91.48$    &    $7.4e\!-\!05$\\
\midrule
\textsf{WikiConflict$16|V|$}   & $1$\,M    &   $67$\,M      & $0.62$ &$129.04$    &    $3.4e\!-\!05$\\
\textsf{Epinions$16|V|$}       & $2$\,M    & $23$\,M    & $0.17$    &    $75.99$    &    $9.5e\!-\!06$\\
\bottomrule
\end{tabular}
}
\end{table}

\spara{Implementation.}
All methods, with the exception of algorithm {\sc FOCG} (details in Section~\ref{sec:exp2}), are implemented in Python (v. 2.7.15) and compiled by Cython.
The experiments run on a machine equipped with Intel Xeon CPU at 2.1GHz and 128GB RAM.\footnote{Code and datasets available at \href{https://github.com/egalimberti/polarized_communities}{github.com/egalimberti/polarized\_communities}.}

\begin{figure}[t!]
\centerline{
\includegraphics[width=1\columnwidth]{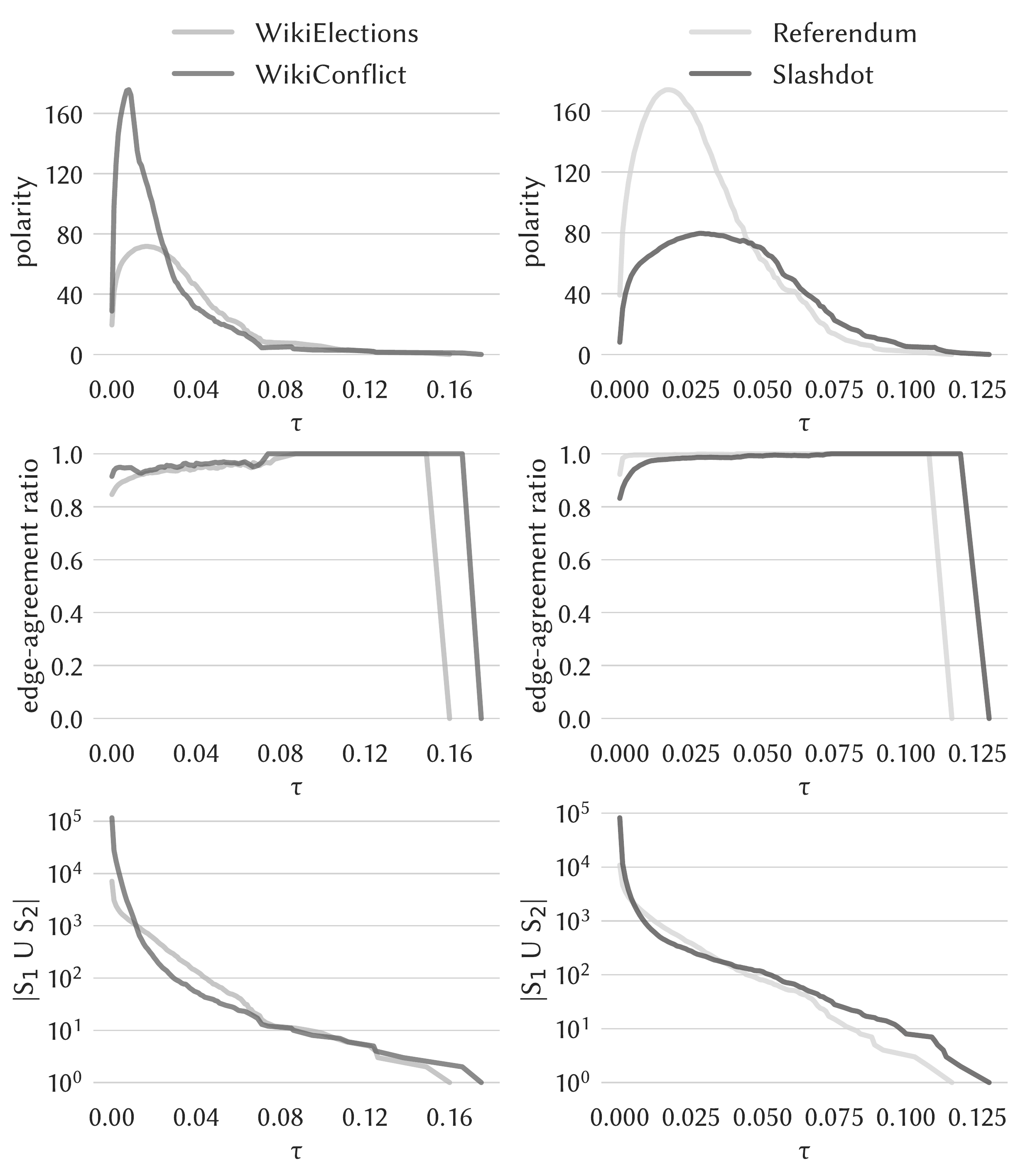}
}
\vspace{-4mm}
\caption{\label{fig:evaluation_tau} Solutions produced by  {\sc E} as a function of $\tau$.}
\end{figure}
\begin{figure}[t!]
\centerline{
\includegraphics[width=1\columnwidth]{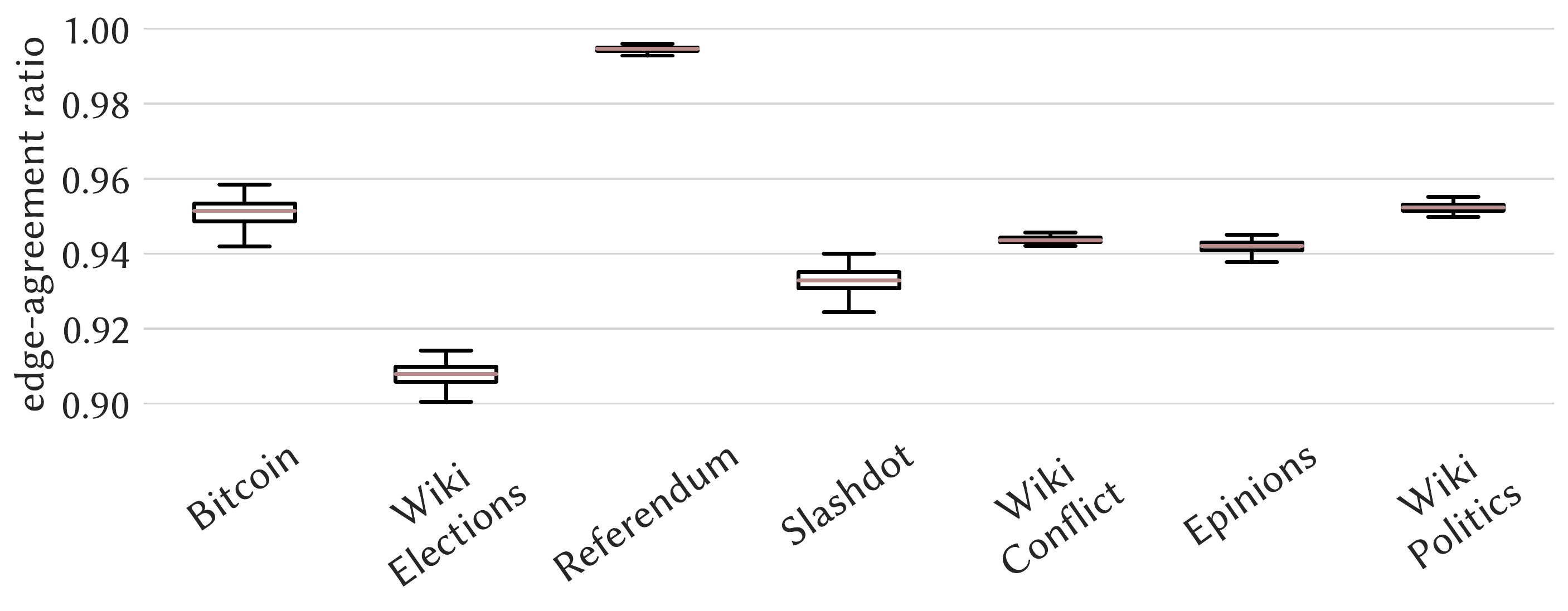}
}
\vspace{-4mm}
\caption{\label{fig:evaluation_re} Edge-agreement ratio of the solutions of {\sc RE}.}
\end{figure}
\subsection{Solutions characterization}\label{sec:exp1}
We first characterize the solutions discovered by our methods \eigensign (for short {\sc E}) and \randalgo ({\sc RE}), and we show how the tweaks described in Section~\ref{subsection:practical} enhance their flexibility in producing a wider variety of results.
In particular, algorithm {\sc E} evaluates the threshold $\tau$ for each $|v_i|$ discretized at the third decimal digit.
This operation is carried out efficiently, since $\vec{v}$ is computed only once regardless of the number of evaluated values of $\tau$.
On the other hand, algorithm {\sc RE} employs $\|v\|_1$ as multiplicative factor, therefore the probabilities are modified to be $\min\{1, \|\vec{v}\|_1 |v_i|\}$.
In the following, we refer to the two communities included in the solutions as $S_1$ and $S_2$, namely the subsets of vertices that are assigned with $1$ and $-1$, respectively, by the solution vector $\vec{x}$.

Figure~\ref{fig:evaluation_tau} shows how the solutions returned by algorithm {\sc E} are affected by parameter $\tau$ in terms of \score, edge-agreement ratio (i.e., the portion of edges in the solution that comply with the polarized structure), and size on four datasets.
In all of them, the three measures follow very similar trends.
The highest \score is achieved at about a fourth of the domain of $\tau$, when most of the neutral vertices are discarded.
The edge-agreement ratio, instead, is consistently close or equal to $1$: the solutions have a coherent polarized structure regardless of the chosen $\tau$.
Finally, as expected, the number of vertices included in the solutions decreases as $\tau$ grows, and presents a substantial decay at the beginning of the domain.
Therefore, parameter $\tau$ is a powerful enhancement that allows algorithm {\sc E} to be tuned to return the most suitable solution for the domain under analysis.

For algorithm {\sc RE}, due to the randomness, we report the best solution with respect to \score out of $100$ runs.
We do the same for the baseline {\sc LS}, that we introduce in Section~\ref{sec:exp2}.
Figure~\ref{fig:evaluation_re} shows the boxplots of the edge-agreement ratio over the larger datasets.
It has significant values in all cases, above $0.9$, and is stable among the different executions.
Polarity and solution size for all datasets are reported in Figure~\ref{fig:evaluation}.
For such measures, we do not show boxplots as they are highly dependent on the specific dataset and very stable over different runs: their index of dispersion is lower than $0.01$ and $3.2e\!-\!05$, respectively, for all datasets.
This confirms that algorithm {\sc RE} is very stable and does not require multiple executions to identify high-quality solutions.

\subsection{Comparative evaluation}\label{sec:exp2}
We next compare algorithms {\sc E} and {\sc RE} against non-trivial baselines inspired by methods proposed in the literature for different yet related problems.

\spara{FOCG.}
The first method we compare to, whose objective is to find $k$ \emph{oppositive cohesive groups} (i.e., $k$-OCG) in signed networks, is taken from~\cite{chu2016finding}.
Algorithm {\sc FOCG} detects $p$ different $k$-OCG structures within the input signed network, among which we elect the one having highest polarity as the ultimate solution to our problem.
We setup the algorithm with the default configuration (i.e., $\alpha = 0.3$ and $\beta = 50$) and $k = 2$.
The code is provided by the authors.

\spara{Greedy.}
Our second baseline is inspired by the 2-approximation algorithm for \emph{densest subgraph}~\cite{charikar2000greedy}.
Algorithm {\sc Greedy} (for short {\sc G}), iteratively removes the vertex minimizing the difference between the number of positive adjacent edges and the number of negative adjacent edges, up to when the graph is empty.
At the end, it returns the subgraph having the highest polarity among all subgraphs visited during its execution.
The assignment of the vertices to the clusters is guided by the sign of the components of the eigenvector~$\vec{v}$, corresponding to the largest eigenvalue of $A$.

\spara{Bansal.}
A different approach, motivated by the strong similarity to our setting, is inspired by Bansal's 3-approximation algorithm for \twocc on complete signed graphs~\cite{bansal2004correlation}.
For each vertex $u \in V$,
this algorithm, which we refer to as {\sc Bansal} (for short {\sc B}), identifies $u$ together with the vertices sharing a positive edge with $u$ as one cluster, and the vertices sharing a negative edge as the other.
Of these $|V|$ possible solutions, it returns the one maximizing polarity.

\spara{LocalSearch.}
Finally, we consider a local search approach ({\sc LocalSearch}, for short {\sc LS}), guided by our objective function.
Algorithm {\sc LS} starts from a set of vertices chosen at random; at each iteration, it adds (removes) to (from) the current solution the vertex that maximizes the gain in terms of polarity, and finally terminates when the gain of moving any vertex is lower than $0.2$.
Also for this algorithm, the assignment of the vertices to the clusters is guided by the signs of $\vec{v}$.

\begin{figure}[t]
\centerline{
\includegraphics[width=1\columnwidth]{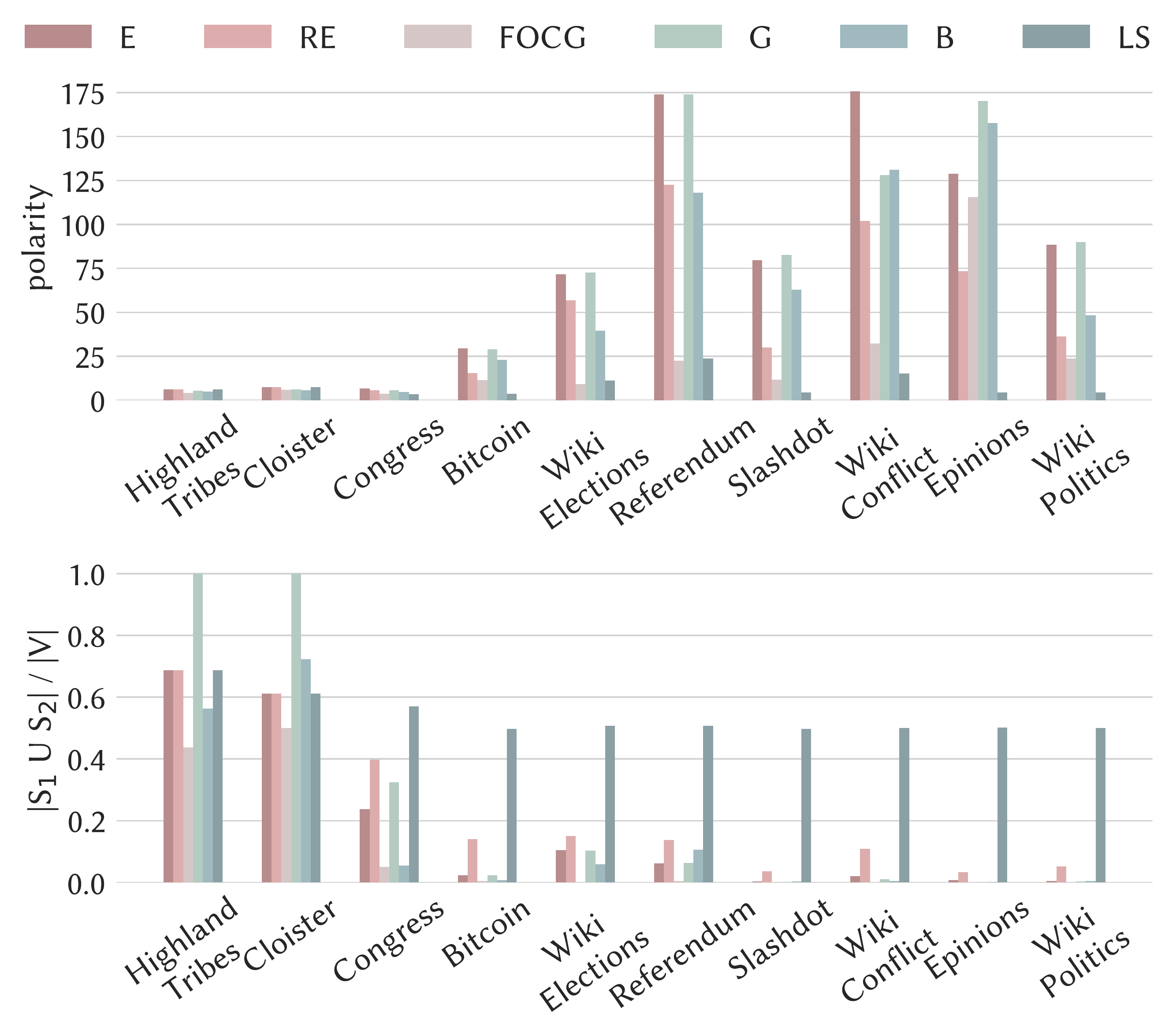}
}
\vspace{-4mm}
\caption{\label{fig:evaluation} Polarity and solution size (normalized) of the proposed algorithms and baselines.}
\end{figure}

Figure~\ref{fig:evaluation} reports the achieved values of  \score\ for all compared algorithms on all datasets, as well as the size (normalized by $|V|$) of the solutions returned.
%
%
%
In most of the cases, algorithm {\sc E} results the be the most competitive method with respect to \score; on the other hand, algorithm {\sc RE} is able to return solutions of high \score for the small-sized datasets.
Algorithm {\sc FOCG} is instead not competitive since its solutions are of extremely small size (note that the numerator of our objective can be up to quadratic in the size of the denominator, so size matters for reaching high \score).
Algorithm {\sc G} has, in general, \score comparable to algorithm {\sc E}, slightly higher in a few cases (with the exception of \textsf{WikiConflict}, in which algorithm {\sc E} clearly outperforms algorithm {\sc G}).
However, it must be noted that algorithm {\sc G} often returns a very dense subgraph as one of the two communities, leaving the second community totally empty, which is, of course, undesirable in our context.
Algorithms {\sc B} and {\sc LS}, instead, exhibit weak performance in terms of \score: their search spaces strongly depend on the neighborhood structure of the vertices (for {\sc B}), or on the random starting sets (for {\sc LS}).
About the solution size, all methods, with the exception of algorithms {\sc FOCG} and {\sc LS}, return solutions of reasonable dimension with respect to the number of vertices of the networks.
Excluding the small empirical datasets (i.e., \textsf{HighlandTribes}, \textsf{Cloister}, and \textsf{Congress}), the size of the solutions is below $20\%$ of the input.

\begin{figure}[t!]
\centerline{
\includegraphics[width=1\columnwidth]{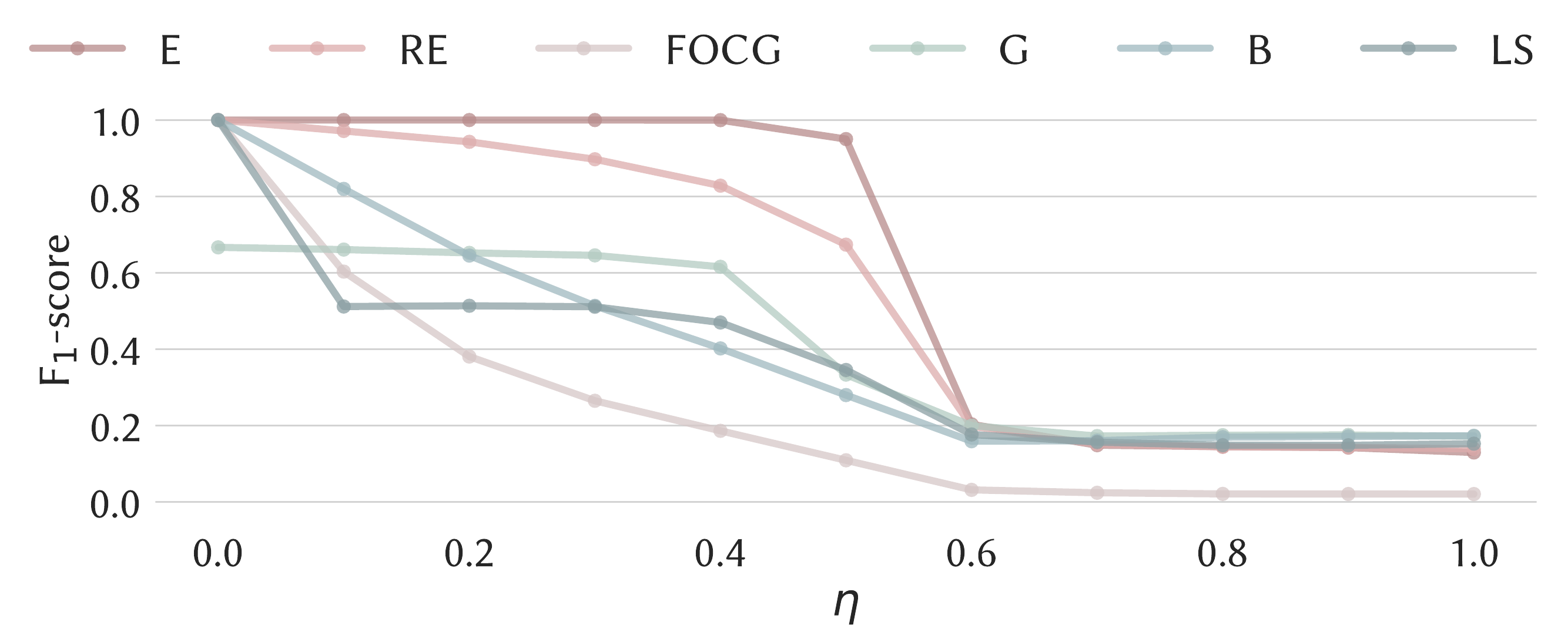}
}
\vspace{-4mm}
\caption{\label{fig:f1_eta} $F_1$-score as a function of the noise parameter $\eta$ ($n_c = 100$, $n_n = 800$).}

\centerline{
\includegraphics[width=1\columnwidth]{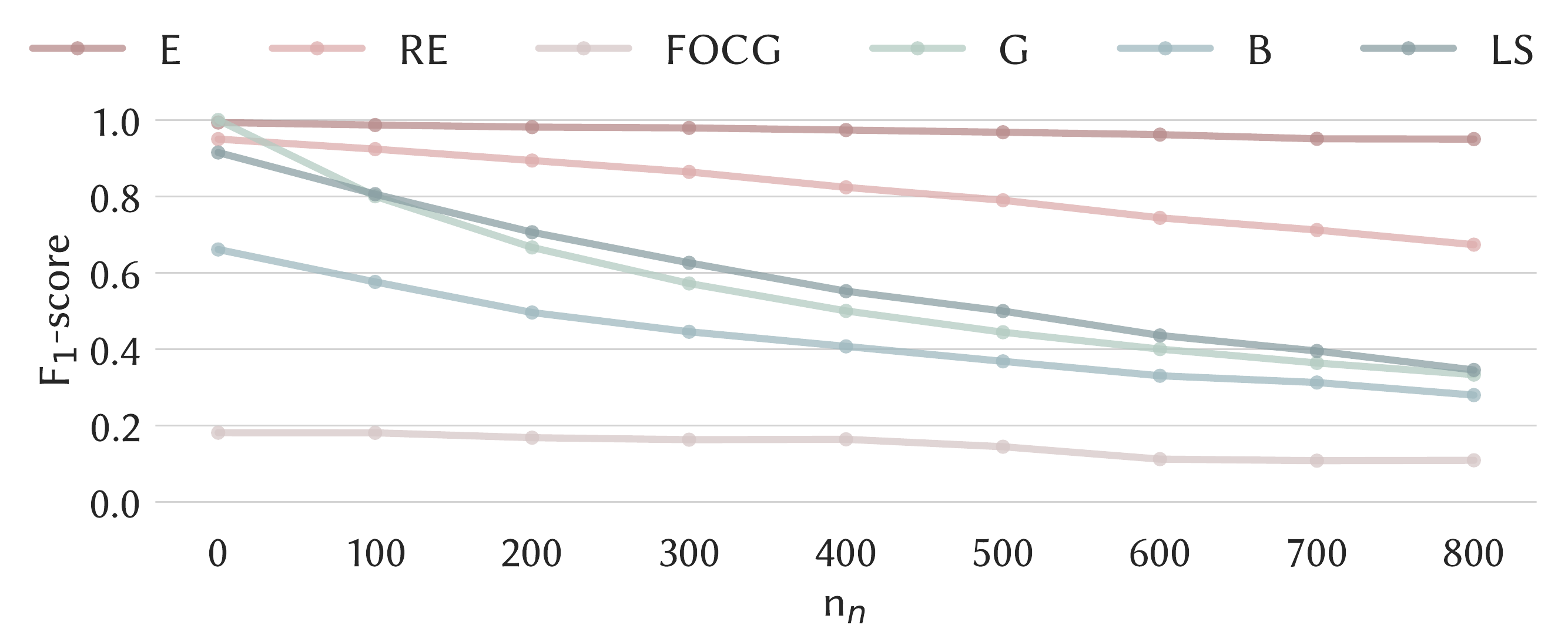}
}
\vspace{-4mm}
\caption{\label{fig:f1_vertices} $F_1$-score as a function of the number noisy vertices $n_n$ ($n_c = 100$, $\eta = 0.5$).}
\end{figure}

\spara{Planted polarized communities.}
In order to better assess the effectiveness of the various algorithms, we test their ability to detect a known planted solution, concealed within varying amounts of noise. For our purposes we create a collection of synthetic signed networks identified by three parameters: the size of each planted polarized community $n_c =  |S_1| = |S_2|$ (for convenience, we consider communities having the same size); the number of noisy vertices external to the two polarized communities $n_n = |V \setminus (S_1 \cup S_2)|$; and, a \emph{noise parameter} $\eta \in [0,1]$ governing the edge density and agreement to the model. In detail:
\squishlist
\item edges inside $S_1$ (respect. $S_2$) exist and are positive with probability $1 - \eta$, exist and are negative with probability $\eta/2$, and do not exist with probability $\eta/2$;
\item edges between  $S_1$ and  $S_2$ exist and are negative with probability $1 - \eta$, exist and are positive with probability $\eta/2$, and do not exist with probability $\eta/2$;
\item all other edges (outside the two polarized communities) exist with probability $\eta$ and have equal probability of being positive or negative.
\squishend

\noindent The higher $\eta$, the less internally dense and polarized the two communities are, and the more connected the noisy vertices are, both between themselves and to the communities. Observe how the case with no noise ($\eta = 0$) corresponds to the ``perfect'' structure.

For each configuration of the parameters, we create $10$ different networks and we report the average $F_1$-score in detecting which vertices belong to $S_1$ (respect. $S_2$) and which ones to $V \setminus (S_1 \cup S_2)$\footnote{For instance \emph{recall} is defined as $(|S_1^* \cap S_1| + |S_2^* \cap S_2|)/ |S_1 \cup S_2|$, where $S_1^*$ ($S_2^*$) denotes the first (second) community returned by the algorithm while $S_1$ ($S_2$) denotes the corresponding ground-truth one.}.

In Figure~\ref{fig:f1_eta} we fix the size of the synthetic network to $1\,000$ ($n_c = 100$, $n_n = 800$) and vary $\eta$.
For $\eta = 0$, all algorithms have, as expected, maximum $F_1$-score with the exception of algorithm {\sc G} that, even in the case without noise, is not able to exactly identify the planted structure.
As expected, as $\eta$ increases, the $F_1$-score decays for all methods; however, our algorithms {\sc E} and {\sc RE} clearly outperform the others.
Figure~\ref{fig:f1_vertices} shows the $F_1$-score varying the number $n_n$ of vertices external to the polarized communities, with fixed $n_c = 100$ and $\eta = 0.5$.
Again algorithms {\sc E} and {\sc RE} stand out, especially {\sc E} that presents $F_1$-score close to the maximum in all cases.
Algorithm {\sc FOCG} has the poorest performance: the small size of its solutions penalizes the recall, which is never greater than $0.1$.

\begin{figure}[t!]
\centerline{
\includegraphics[width=1\columnwidth]{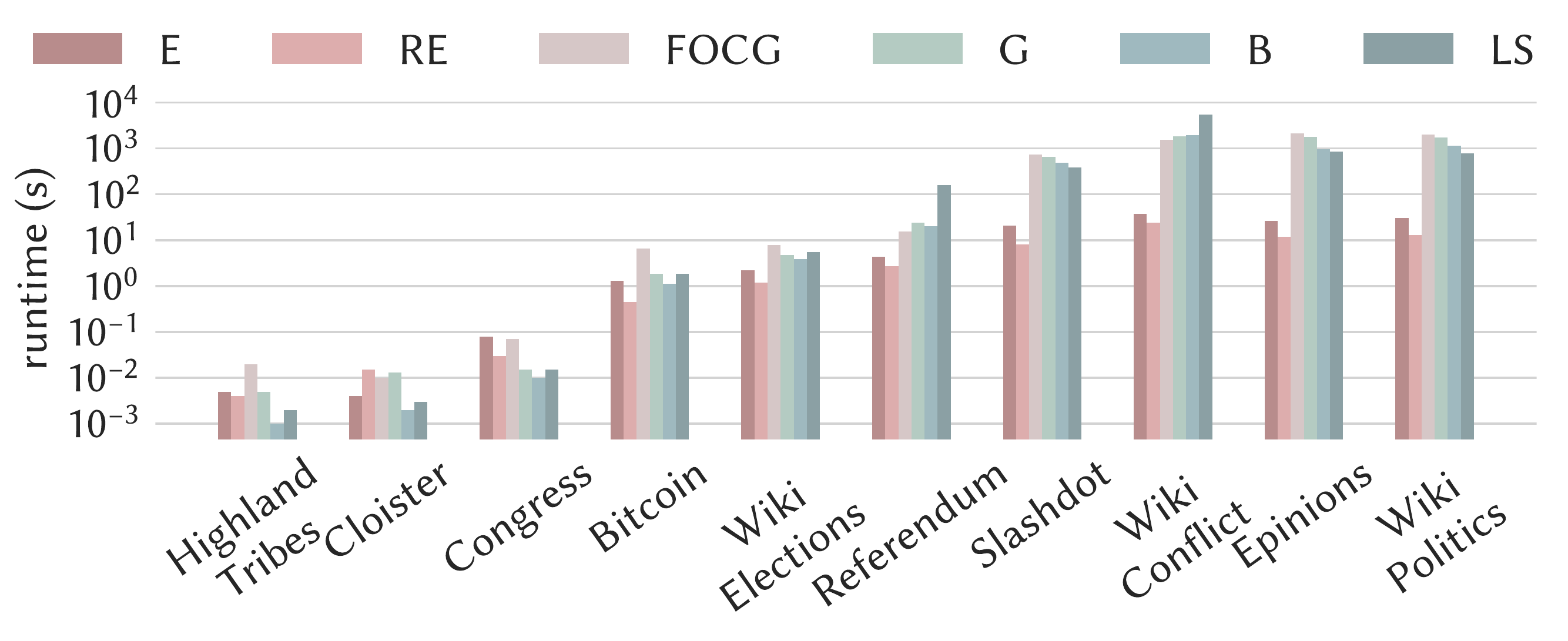}
}
\vspace{-4mm}
\caption{\label{fig:evaluation_runtime} Runtime of the proposed algorithms and baselines.}

\centerline{
\includegraphics[width=1\columnwidth]{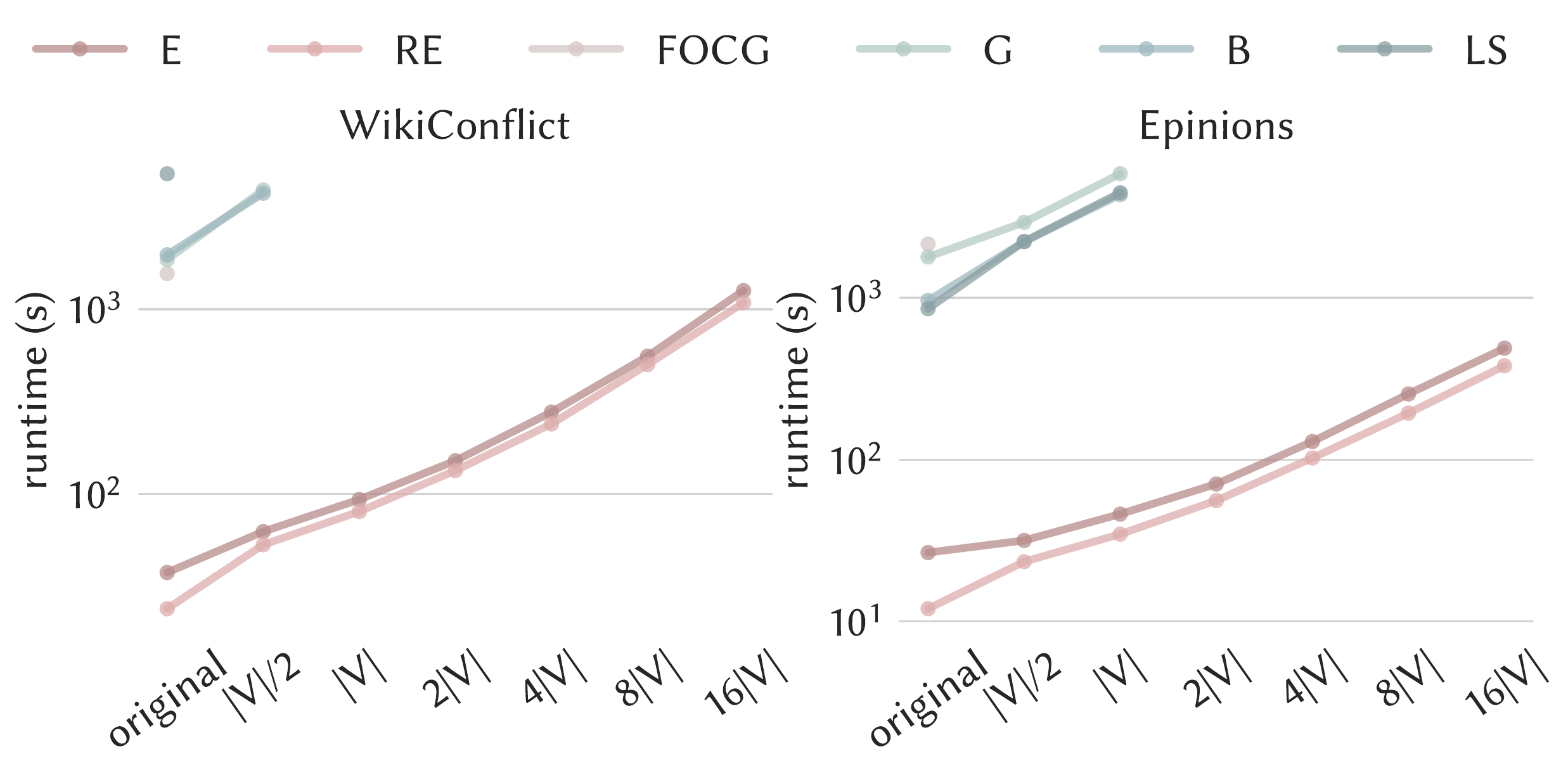}
}
\vspace{-4mm}
\caption{\label{fig:evaluation_scalability} Scalability: runtime of the proposed algorithms and baselines as a function of the number of injected dummy vertices, for \textsf{WikiConflict} and \textsf{Epinions}.}
\end{figure}

\spara{Runtime and scalability.}
Figure~\ref{fig:evaluation_runtime} reports the runtime of all algorithms over all datasets.
Algorithms {\sc E} and {\sc RE}, with their practical enhancements discussed in Section \ref{subsection:practical}, always terminate in less than $40$ seconds.
The runtime of the baselines is instead more than an order of magnitude higher than algorithms {\sc E} and {\sc RE}.

In order to assess the scalability of our methods, we augment two of the larger datasets (i.e., \textsf{WikiConflict} and \textsf{Epinions})
by artificially injecting dummy vertices having a number of randomly-connected edges equal to the average degree of the original network, while maintaining $\rho_-$ (i.e., the ratio of negative edges).
The largest datasets created in this way contain up to $2$\,M vertices and $67$\,M edges (see Table~\ref{tab:datasets} for details).
Note that, as the quantity of noise increases, $\delta$, i.e., the ratio of non-zero elements of the adjacency matrix, decreases.
Nonetheless, $\delta$ differs with respect to the original datasets less than an order of magnitude in both cases, making the following results about scalability significant.

Figure~\ref{fig:evaluation_scalability}, which reports on the $x$-axis the number of dummy vertices added, shows that the runtime of both algorithms {\sc E} and {\sc RE} grows linearly with the number of vertices.
Among the two, algorithm {\sc E} is slightly slower than algorithm {\sc RE} due to the evaluation of multiple values of the threshold $\tau$.
In the worst case, algorithms {\sc E} returns in about $21$ minutes.
On the other hand, the baselines cannot complete each computation within the $10\,000$ seconds timeout that we apply.
No baseline terminates for more than $|V|$ additional dummy vertices on both datasets.
In particular, algorithm {\sc FOCG} is able to handle in reasonable time only the original versions, with no dummy vertices.
It should be noted that algorithm {\sc FOCG} recursively finds a polarized structure, removes the corresponding subgraph, and repeats the process on the remaining vertices.
While each one of these iterations runs efficiently, most of the found structures are too small to be of interest in our setting.
Thus, it is necessary to allow the algorithm to complete many of such iterations in order to find interesting solutions.

\subsection{Case study: political debate}\label{sec:exp3}
We finally analyze the solution extracted by algorithm {\sc RE} from \textsf{Referendum} to show tangible benefits of our problem formulation and algorithms in identifying the two most polarized communities in a signed network modeling political debates.
The \textsf{Referendum} dataset includes Twitter data about the Italian Constitutional Referendum held on December 4, 2016 (more information about the Referendum can be found at this \href{https://en.wikipedia.org/wiki/2016_Italian_constitutional_referendum}{link}).
The original data seed consists of about $1$\,M tweets posted between November 24 and December 7, 2016, extended by collecting retweets, quotes, and replies.
The users ($10\,884$ in total) are annotated with a stance about their outlook towards the Referendum as favorable ($5\,137$), against ($1\,510$), or none ($4\,237$) when the stance cannot be inferred.
An interaction (edge) is considered negative if occurred between two users (vertices) of different stances, and is positive otherwise, i.e., we treat ``none'' users as neutral, in agreement with both favorable and against users.

 The solution output by algorithm {\sc RE} consists of two communities of $27$ and $1\,558$ users, accounting for 14\% of the overall user set.
Both communities have more than 99\% of positive edges within and 74\% of negative edges in-between, and thus, are highly polarized.
Interestingly, all the $27$ users of the smaller community are classified as favorable to the Referendum, while the users in the larger community as against (75\%) or ``none'' (24\%), with the exception of $3$ favorables.
Moreover, the vertices in the solution have, on average, $183.12$ adjacent edges compared to the average $22.85$ contacts of the vertices outside, meaning that the solution identifies the ``core'' of the controversies, i.e., a set of intensely debating users about the Referendum.
These results provide  evidence of the practical value of our problem formulation and algorithms to identify two communities that are polarized about a certain topic.

\section{Conclusions and Future Work}
\label{sec:conclusions}
Detecting extremely polarized communities might enable fine-grained analysis of controversy in social networks, as well as open the door to interventions aimed at reducing it \cite{garimella2017reducing}.
As a step in this direction, in this paper we introduce the \ourproblemfull\ problem,
which requires finding two communities (subsets of the network vertices)
such that within communities there are mostly positive edges
while across communities there are mostly negative edges.
We prove that the proposed problem is \NP-hard and devise two efficient algorithms with provable approximation guarantees. Through an extensive set of experiments on a wide variety of real-world networks, we show how the proposed objective function can be optimized to reveal polarized communities. Our experiments confirm that our algorithms are more accurate, faster, and more scalable than non-trivial baselines.

This work opens several enticing avenues for further inquiry. Some questions follow immediately from our theoretical results. What are the approximation capabilities of \eigensign in signed networks? Can we improve the $\sqrt n$ factor of \randalgo, e.g., by multiplying the probability vectors by $\|v\|_1$ or some other factor? Finally, it would be interesting to extend the problem to detect an arbitrary number of communities.

The application of the proposed algorithms to real-world networks with positive and negative relationships can have implications in computational social science problems. For instance, understanding opinion shifts in data streaming from social media sources can be investigated in terms of polarized communities. Opinions shared within vertices (individuals) belonging to the same community are likely to be reinforced after different interactions; discussions within individuals with antagonistic perspectives may result in both opinion shifts and controversy amplification. The identification of a subgraph made of vertices belonging to \ourproblemfull\ may lead to novel ways to discover the basic laws behind opinion shift dynamics. Thus, it would be interesting to study extensions of the \ourproblemfull\ problem in the setting of temporal networks.



\begin{thebibliography}{10}
\enlargethispage*{\baselineskip}

\bibitem{ailon2008aggregating}
N.~Ailon, M.~Charikar, and A.~Newman.
\newblock Aggregating inconsistent information: ranking and clustering.
\newblock {\em Journal of the ACM (JACM)}, 55(5):23, 2008.


\bibitem{akoglu2014quantifying}
L.~Akoglu.
\newblock Quantifying political polarity based on bipartite opinion networks.
\newblock In {\em ICWSM}, 2014.

\bibitem{anchuri2012communities}
P.~Anchuri and M.~Magdon-Ismail.
\newblock Communities and balance in signed networks: A spectral approach.
\newblock In {\em ASONAM}, 2012.

\bibitem{antal2006}
T.~Antal, P.~Krapivsky, and S.~Redner.
\newblock Social balance on networks: The dynamics of friendship and enmity.
\newblock {\em Physica D}, 224(130), 2006.

\bibitem{baldassarri2007dynamics}
D.~Baldassarri and P.~Bearman.
\newblock Dynamics of political polarization.
\newblock {\em American sociological review}, 72(5):784--811, 2007.

\bibitem{bansal2004correlation}
N.~Bansal, A.~Blum, and S.~Chawla.
\newblock Correlation clustering.
\newblock {\em Machine learning}, 56(1-3):89--113, 2004.

\bibitem{beigi2016signed}
G.~Beigi, J.~Tang, and H.~Liu.
\newblock Signed link analysis in social media networks.
\newblock In {\em ICWSM}, 2016.

\bibitem{brundidge2010encountering}
J.~Brundidge.
\newblock Encountering `difference' in the contemporary public sphere: The
  contribution of the internet to the heterogeneity of political discussion
  networks.
\newblock {\em Journal of Communication}, 60(4):680--700, 2010.

\bibitem{cartwright1956structural}
D.~Cartwright and F.~Harary.
\newblock Structural balance: a generalization of heider's theory.
\newblock {\em Psychological review}, 63(5):277, 1956.

\bibitem{cesa2012correlation}
N.~Cesa-Bianchi, C.~Gentile, F.~Vitale, and G.~Zappella.
\newblock A correlation clustering approach to link classification in signed
  networks.
\newblock In {\em COLT}, 2012.

\bibitem{charikar2000greedy}
M.~Charikar.
\newblock Greedy approximation algorithms for finding dense components in a
  graph.
\newblock In {\em International Workshop on Approximation Algorithms for
  Combinatorial Optimization}, pages 84--95, 2000.

\bibitem{chiang2012scalable}
K.-Y. Chiang, J.~J. Whang, and I.~S. Dhillon.
\newblock Scalable clustering of signed networks using balance normalized cut.
\newblock In {\em CIKM}, 2012.

\bibitem{choi2010identifying}
Y.~Choi, Y.~Jung, and S.-H. Myaeng.
\newblock Identifying controversial issues and their sub-topics in news
  articles.
\newblock In {\em Pacific-Asia Workshop on Intelligence and Security
  Informatics}, 2010.

\bibitem{chu2016finding}
L.~Chu, Z.~Wang, J.~Pei, J.~Wang, Z.~Zhao, and E.~Chen.
\newblock Finding gangs in war from signed networks.
\newblock In {\em KDD}, 2016.

\bibitem{coleman2008local}
T.~Coleman, J.~Saunderson, and A.~Wirth.
\newblock A local-search 2-approximation for 2-correlation-clustering.
\newblock In {\em ESA}, 2008.

\bibitem{conover2011predicting}
M.~D. Conover, B.~Gon{\c{c}}alves, J.~Ratkiewicz, A.~Flammini, and F.~Menczer.
\newblock Predicting the political alignment of twitter users.
\newblock In {\em SocialCom/PASSAT}, 2011.


\bibitem{cribari2000note}
F. Cribari-Neto, N.L. Garcia, and K. Vasconcellos.
\newblock A note on inverse moments of binomial variates.
\newblock Brazilian Review of Econometrics, 20:2, 2000.


\bibitem{esteban1994measurement}
J.-M. Esteban and D.~Ray.
\newblock On the measurement of polarization.
\newblock {\em Econometrica: Journal of the Econometric Society}, pages
  819--851, 1994.

\bibitem{feldman2014mutual}
L.~Feldman, T.~A. Myers, J.~D. Hmielowski, and A.~Leiserowitz.
\newblock The mutual reinforcement of media selectivity and effects: Testing
  the reinforcing spirals framework in the context of global warming.
\newblock {\em Journal of Communication}, 64(4):590--611, 2014.

\bibitem{gao2016detecting}
M.~Gao, E.-P. Lim, D.~Lo, and P.~K. Prasetyo.
\newblock On detecting maximal quasi antagonistic communities in signed graphs.
\newblock {\em Data mining and knowledge discovery}, 30(1):99--146, 2016.

\bibitem{garimella2017reducing}
K.~Garimella, G.~De~Francisci~Morales, A.~Gionis, and M.~Mathioudakis.
\newblock Reducing controversy by connecting opposing views.
\newblock In {\em WSDM}, 2017.

\bibitem{garimella2018quantifying}
K.~Garimella, G.~D.~F. Morales, A.~Gionis, and M.~Mathioudakis.
\newblock Quantifying controversy on social media.
\newblock {\em ACM Transactions on Social Computing}, 1(1):3, 2018.

\bibitem{garrett2014partisan}
K.~Garrett and N.~J. Stroud.
\newblock Partisan paths to exposure diversity: Differences in pro-and
  counter-attitudinal news consumption.
\newblock {\em Journal of Communication}, 64(4):680--701, 2014.

\bibitem{giotis2006correlation}
I.~Giotis and V.~Guruswami.
\newblock Correlation clustering with a fixed number of clusters.
\newblock In {\em  SODA}, 2006.

\bibitem{graells2014people}
E.~Graells-Garrido, M.~Lalmas, and D.~Quercia.
\newblock People of opposing views can share common interests.
\newblock In {\em WWW}, 2014.

\bibitem{guha2004propagation}
R.~Guha, R.~Kumar, P.~Raghavan, and A.~Tomkins.
\newblock Propagation of trust and distrust.
\newblock In {\em WWW}, 2004.

\bibitem{harary1953notion}
F.~Harary.
\newblock On the notion of balance of a signed graph.
\newblock {\em The Michigan Mathematical Journal}, 2(2):143--146, 1953.


\bibitem{hou2003laplacian}
Y.~Hou, J.~Li, and Y.~Pan.
\newblock On the {L}aplacian eigenvalues of signed graphs.
\newblock {\em Linear and Multilinear Algebra}, 51(1):21--30, 2003.

\bibitem{hou2005bounds}
Y.~P. Hou.
\newblock Bounds for the least {L}aplacian eigenvalue of a signed graph.
\newblock {\em Acta Mathematica Sinica}, 21(4):955--960, 2005.

\bibitem{kunegis2009slashdot}
J.~Kunegis, A.~Lommatzsch, and C.~Bauckhage.
\newblock The slashdot zoo: mining a social network with negative edges.
\newblock In {\em WWW}, 2009.

\bibitem{kunegis2010spectral}
J.~Kunegis, S.~Schmidt, A.~Lommatzsch, J.~Lerner, E.~W. De~Luca, and
  S.~Albayrak.
\newblock Spectral analysis of signed graphs for clustering, prediction and
  visualization.
\newblock In {\em SDM}, 2010.

\bibitem{lai2018stance}
M.~Lai, V.~Patti, G.~Ruffo, and P.~Rosso.
\newblock Stance evolution and twitter interactions in an italian political
  debate.
\newblock In {\em NLDB}, 2018.

\bibitem{leskovec2010predicting}
J.~Leskovec, D.~Huttenlocher, and J.~Kleinberg.
\newblock Predicting positive and negative links in online social networks.
\newblock In {\em WWW}, 2010.

\bibitem{leskovec2010signed}
J.~Leskovec, D.~Huttenlocher, and J.~Kleinberg.
\newblock Signed networks in social media.
\newblock In {\em SIGCHI}, 2010.

\bibitem{li2013influence}
Y.~Li, W.~Chen, Y.~Wang, and Z.-L. Zhang.
\newblock Influence diffusion dynamics and influence maximization in social
  networks with friend and foe relationships.
\newblock In {\em WSDM}, 2013.

\bibitem{liao2014can}
Q.~V. Liao and W.-T. Fu.
\newblock Can you hear me now?: mitigating the echo chamber effect by source
  position indicators.
\newblock In {\em CSCW}, 2014.

\bibitem{liao2014expert}
Q.~V. Liao and W.-T. Fu.
\newblock Expert voices in echo chambers: effects of source expertise
  indicators on exposure to diverse opinions.
\newblock In {\em SIGCHI}, 2014.

\bibitem{lo2011mining}
D.~Lo, D.~Surian, K.~Zhang, and E.-P. Lim.
\newblock Mining direct antagonistic communities in explicit trust networks.
\newblock In {\em CIKM}, 2011.

\bibitem{ma2009learning}
H.~Ma, M.~R. Lyu, and I.~King.
\newblock Learning to recommend with trust and distrust relationships.
\newblock In {\em RecSys}, 2009.

\bibitem{marshall1979inequalities}
A.~W. Marshall, I.~Olkin, and B.~C. Arnold.
\newblock {\em Inequalities: theory of majorization and its applications},
  volume 143.
\newblock Springer, 1979.

\bibitem{marvel2009}
S.~A. Marvel, S.~H. Strogatz, and J.~M. Kleinberg.
\newblock The energy landscape of social balance.
\newblock {\em Physical Review Letters}, 103(19), 2009.

\bibitem{mejova2014controversy}
Y.~Mejova, A.~X. Zhang, N.~Diakopoulos, and C.~Castillo.
\newblock Controversy and sentiment in online news.
\newblock {\em CJ'14: Computation+Journalism Symposium}, 2014.

\bibitem{munson2013encouraging}
S.~A. Munson, S.~Y. Lee, and P.~Resnick.
\newblock Encouraging reading of diverse political viewpoints with a browser
  widget.
\newblock In {\em ICWSM}, 2013.

\bibitem{popescu2010detecting}
A.-M. Popescu and M.~Pennacchiotti.
\newblock Detecting controversial events from twitter.
\newblock In {\em CIKM}, 2010.

\bibitem{puleo2015correlation}
G.~J. Puleo and O.~Milenkovic.
\newblock Correlation clustering with constrained cluster sizes and extended
  weights bounds.
\newblock {\em SIAM Journal on Optimization}, 25(3), 2015.

\bibitem{shamir2004cluster}
R.~Shamir, R.~Sharan, and D.~Tsur.
\newblock Cluster graph modification problems.
\newblock {\em Discrete Applied Mathematics}, 144(1-2):173--182, 2004.

\bibitem{shi2000normalized}
J.~Shi and J.~Malik.
\newblock Normalized cuts and image segmentation.
\newblock {\em IEEE Transactions on pattern analysis and machine intelligence},
  22(8):888--905, 2000.

\bibitem{swamy2004correlation}
C.~Swamy.
\newblock Correlation clustering: maximizing agreements via semidefinite
  programming.
\newblock In {\em SODA}, 2004.

\bibitem{symeonidis2010transitive}
P.~Symeonidis, E.~Tiakas, and Y.~Manolopoulos.
\newblock Transitive node similarity for link prediction in social networks
  with positive and negative links.
\newblock In {\em RecSys}, 2010.

\bibitem{tang2016node}
J.~Tang, C.~Aggarwal, and H.~Liu.
\newblock Node classification in signed social networks.
\newblock In {\em SDM}, 2016.

\bibitem{tang2016recommendations}
J.~Tang, C.~Aggarwal, and H.~Liu.
\newblock Recommendations in signed social networks.
\newblock In {\em WWW}, 2016.

\bibitem{tang2016survey}
J.~Tang, Y.~Chang, C.~Aggarwal, and H.~Liu.
\newblock A survey of signed network mining in social media.
\newblock {\em ACM Computing Surveys (CSUR)}, 49(3):42, 2016.

\bibitem{victor2011trust}
P.~Victor, C.~Cornelis, M.~De~Cock, and A.~M. Teredesai.
\newblock Trust-and distrust-based recommendations for controversial reviews.
\newblock {\em IEEE Intelligent Systems}, 26(1), 2011.

\bibitem{vydiswaran2015overcoming}
V.~Vydiswaran, C.~Zhai, D.~Roth, and P.~Pirolli.
\newblock Overcoming bias to learn about controversial topics.
\newblock {\em Journal of the Association for Information Science and
  Technology}, 66(8):1655--1672, 2015.

\bibitem{wojcieszak2009online}
M.~Wojcieszak and D.~Mutz.
\newblock Online groups and political discourse: Do online discussion spaces
  facilitate exposure to political disagreement?
\newblock {\em Journal of communication}, 59(1):40--56, 2009.

\end{thebibliography}

\appendix

\section{Appendix}
\label{sec:appendix}

\spara{Hardness (Theorem~\ref{th:hard}).}
In this section, we refer to a solution of \ourproblem\ as $S_1,S_2$,
which denote the subsets of vertices that are assigned a $1$ or a $-1$, respectively,
in the solution vector $\vec{x}$. Given a vertex $v \in V$ and a subset of vertices $S\subseteq V$,
we use $d_+(v,S)$ (respectively $d_-(v,S)$)
to denote the number of `$+$' edges (respectively `$-$' edges) connecting $v$ to other vertex in~$S$.

We exploit the following result in our proof.
It can be easily verified by examining the behavior of the cost functions
when moving one vertex from one set to the other, so we omit the proof.
\begin{proposition}
\label{prop:prob_eq}
If we require $S_1\cup S_2=V$,
problem \twoccfull is equivalent to \twocc,
i.e.,  their optimal solutions are the same.
\end{proposition}

We now prove that \ourproblem\ is \NP-hard by reduction from \twocc,
which has been shown to be \NP-hard by Shamir~\emph{et~al.}~\cite{shamir2004cluster}.

\begin{proof}[Proof of Theorem~\ref{th:hard}]
Given a graph $\tilde G=(\tilde V, \tilde E)$ with $n$ vertices as instance of \twocc,
we construct a graph $G=(V,E)$ to be an instance of \ourproblem\ as follows.
For every vertex in $\tilde V$ we create a corresponding vertex in $V$,
and for every edge in  $\tilde E$ we add an edge in $E$ between the corresponding vertices in $\tilde V$,
and having the same sign.
Furthermore, for every vertex $v$ in $\tilde V$ we introduce a clique of $m>3n$ vertices (and positive edges)
and a `$+$' edge between $v$ and every vertex in the clique.
The strategy to prove hardness is the following. We first restrict
ourselves to {\em complete} solutions of \ourproblem (i.e. $S_1\cup
S_2=V$), which can of course be mapped to solutions of \twocc.
We prove that if one such complete solution $S_1, S_2$ optimizes \ourproblem,
the corresponding solution of \twocc is also a maximizer.
Second, we show that any optimal solution of the constructed instance of \ourproblem\ is complete.

\enlargethispage{\baselineskip}

We denote the objective of the problems \twocc and \ourproblem,
on instances $\tilde G$ and $G$,
by $\wma$ and $\wsqp$, respectively.
We consider a solution $\tilde S_1, \tilde S_2$ of \twocc, and
a solution $S_1,S_2$ of \ourproblem,
such that $\tilde S_1 \subseteq S_1$ and $\tilde S_2 \subseteq S_2$.
Let us first restrict our attention to complete solutions of \ourproblem. Observe that
\begin{align}
\wsqp(S_1,S_2) = & \, \frac{1}{n+nm} \left( \wma(\tilde S_1, \tilde S_2) - D(\tilde S_1,\tilde S_2) \right) \nonumber
\\
& + \frac{1}{n+nm} \left( |S_1|m + |S_2|m + n {m\choose 2} \right), \nonumber
\end{align}
where $D(\tilde S_1,\tilde S_2)=\sum_{v \in \tilde S_1}d_+(v,\tilde S_2) + \sum_{v \in \tilde S_2}d_+(v,\tilde S_2) + d_-(v,\tilde S_1) + d_-(v,\tilde S_2)$,
that is, the sum of disagreements in the resulting clustering.
Note that $\wma(\tilde S_1, \tilde S_2) - D(\tilde S_1,\tilde S_2)$ is exactly the objective of the \twoccfull problem.
In other words, the obective of \ourproblem on $G$
is proportional to the objective of \twoccfull on $\tilde G$ plus a constant.
By Proposition \ref{prop:prob_eq}, the first part of the proof is complete.

We now consider a complete solution $S_1, S_2$ and show that removing vertices leads to no further improvement.
Suppose we remove a set $R$ of $r$ vertices from the solution. We want to show
\[
  \frac{\nu (\wsqp(S_1,S_2))-\Delta(R)}{n+nm-r} < \frac{\nu(\wsqp(S_1,S_2))}{n+nm},
\]
where $\nu (\wsqp(S_1,S_2 ))$ is the numerator of $\wsqp(S_1,S_2)$,
and $\Delta(R)$ is the net change after removing the vertices in $R$
(i.e., the number of agreements minus disagreements that are removed).
Equivalently, we want to show
$\Delta(R)(n+nm) > r\nu (\wsqp(S_1,S_2))$.
We first consider that the removed vertices are in $\tilde V$. Observe that
\[
  \Delta(R) \geq rm-{r\choose 2} - r(n-r),
\]
%
\[
    \nu(\wsqp(S_1,S_2 )) \leq {n\choose 2} + nm + n{m\choose 2}.
\]
This upper bound holds because the right hand side
simply counts all possible `$+$' and `$-$' edges,
the edges between each actual vertex and its clique,
and the edges within cliques.
It is therefore sufficient to show
\begin{align*}
rm - rn + r^2-{r\choose 2} >  r\frac{  {n\choose 2} + nm + n{m\choose 2} }{n+nm}.
\end{align*}
After some manipulations and relaxing the condition to remove the dependence on $r$, we arrive at the following sufficient condition:
\[
  \left (m-n\right )(n+nm) > {n\choose 2} + nm + n{m\choose 2},
\]
which holds for $m > 3n$. The case where the removed vertices are not in $\tilde V$ can be analyzed in the same manner.
We have shown that we can reduce an instance of \twocc to a polynomially-sized instance of \ourproblem.
\end{proof}


\spara{Tight example for \randalgo.}
We consider a complete graph where all edges are positive, except for one Hamiltonian cycle comprised of negative edges. Without loss of generality, we can order the vertices so that the adjacency matrix is
\[
A = \left(\begin{array}{cccccc}
  0 & -1 & 1 & \dots & 1 & -1
\\  -1 & 0 & -1 & 1 & \dots & 1
\\ & &\vdots & \vdots &&
\\  1 & 1 & \dots & -1 & 0 & -1
\\  -1 & 1 & \dots & 1 & -1 & 0  \dots
  \end{array}\right).
\]
That is, matrix $A$ is comprised entirely of ones, save for the subdiagonal and superdiagonal entries, which are -1, and $A_{n1}=A_{1n}=-1$. It is easy to see that a constant vector $\vec{v}$, i.e., satisfying $v_i=v_j$ is an eigenvector of eigenvalue $n-5$. Since $\sum_i\lambda_i^2(A)=\|A\|_F^2=n(n-1)$, the eigenvalue $n-5$ will be the largest if
\[
\frac{n(n-1)}{2} < (n-5)^2,
\]
which holds for $n>16$. Note that $\sqrt n\vec{v}$ is a feasible solution for \ourproblem.
We now show that \randalgo attains a value of $\bigTheta(\sqrt n)$.

We first rely on Equality (\ref{eq:expected_expansion}) to obtain the following:
\begin{align*}
   &\mathbb E \left [\frac{\vec{x}^TA\vec{x}}{\vec{x}^T\vec{x}} \right ] \nonumber
  \\ & = \sum_{k=1}^n \frac{1}{k} \sum_{i\neq j}A_{ij}s_{ij}Pr(x_i=1,x_j=1)Pr(\vec{x}^T\vec{x}=k|x_i=1,x_j=1) \nonumber
  \\ & = \sum_{k=1}^n \frac{1}{k} \sum_{i\neq j}A_{ij}v_iv_jPr(\vec{x}^T\vec{x}=k|x_i=1,x_j=1) \nonumber
\end{align*}

Now, observe that given $k$, $Pr(\vec{x}^T\vec{x}=k|x_h=1,x_l=1)$ is constant for all $i\neq j$. Thus, for arbitrary $h,l$,
\begin{align}
  \label{eq:expected_tight}
   \mathbb E \left [\frac{\vec{x}^TA\vec{x}}{\vec{x}^T\vec{x}} \right ] \nonumber
   &= \sum_{k=1}^n \frac{1}{k} Pr(\vec{x}^T\vec{x}=k|x_h=1,x_l=1) \sum_{i\neq j}A_{ij}v_iv_j
  \\ & = (n-5) \mathbb E\left [\left.\frac{1}{\vec{x}^T\vec{x}} \right|x_h=1,x_l=1\right].
\end{align}

Observe that when all entries of $\vec{v}$ are equal in absolute value, $\vec{x}^T\vec{x}$ is a binomial variable with parameters $(n,|v_i|)=(n,1/\sqrt n)$. Thus, by Jensen's inequality we have
\[\mathbb E\left [\left.\frac{1}{\vec{x}^T\vec{x}} \right|x_h=1,x_l=1\right] \geq \frac{1}{\mathbb E\left [\vec{x}^T\vec{x} |x_h=1,x_l=1\right]} = \bigOmega(1/\sqrt n).
\]
Furthermore, it is known \cite{cribari2000note} that
\[\mathbb E\left [\left.\frac{1}{\vec{x}^T\vec{x}} \right|x_h=1,x_l=1\right] = \bigO(1/\sqrt n).
\]
That is,
\[
\mathbb E\left [\left.\frac{1}{\vec{x}^T\vec{x}} \right|x_h=1,x_l=1\right] = \bigTheta(1/\sqrt n).
\]

Combining this with Equality (\ref{eq:expected_tight}) we get
\[
\mathbb E \left [\frac{\vec{x}^TA\vec{x}}{\vec{x}^T\vec{x}} \right ] = \bigTheta(\sqrt n) = \bigTheta\left(\frac{OPT}{\sqrt n}\right).
\]


\begin{mdframed}[innerbottommargin=3pt,innertopmargin=3pt,innerleftmargin=6pt,innerrightmargin=6pt,backgroundcolor=gray!10,roundcorner=10pt]
\small
\spara{Acknowledgments.} Francesco Bonchi acknowledges support from Intesa Sanpaolo Innovation Center.
Aristides Gionis and Bruno Ordozgoiti were supported by three Academy of Finland projects (286211, 313927, and 317085), and the EC H2020RIA project ``SoBigData'' (654024).
\end{mdframed}

\end{document}